\documentclass[a4paper]{article}

\usepackage[T1]{fontenc}

\usepackage{amsmath, amssymb, amsthm}
\usepackage{geometry}
\usepackage{graphicx}
\usepackage{cancel}
\usepackage{dsfont}

\newtheorem{theo}{Théorème}[section]
\newtheorem{prop}[theo]{Proposition}
\newtheorem{lem}[theo]{Lemme}

\theoremstyle{definition}

\theoremstyle{remark}

\newtheorem*{rem}{Remarque}

\renewcommand\leq{\leqslant}
\renewcommand\geq{\geqslant}
\renewcommand{\Re}{\operatorname{Re}}
\renewcommand{\Im}{\operatorname{Im}}

\newcommand\R{\mathbb R}
\newcommand\C{\mathbb C}

\newcommand{\abs}[1]{\left\lvert#1\right\rvert}
\newcommand{\norme}[1]{\left\lVert#1\right\rVert}
\newcommand{\pare}[1]{\left(#1\right)}
\newcommand{\prods}[2]{\left\langle#1,#2\right\rangle}

\DeclareMathAlphabet{\mathonebb}{U}{bbold}{m}{n}

\numberwithin{equation}{section}

\usepackage[intoc,refpage]{nomencl}
\usepackage[toc,page]{appendix} 
\usepackage{hyperref}
\usepackage{etoolbox}
\apptocmd{\sloppy}{\hbadness 10000\relax}{}{}

\begin{document}

\title{Quasimodes and a lower bound for the local energy decay of the Dirac equation in Schwarzschild-Anti-de Sitter spacetime}
\author{Guillaume \textsc{Idelon-Riton}}

\maketitle

\abstract{We prove the existence of exponentially accurate quasimodes using the square of the Dirac operator on the Schwarzschild-Anti-de Sitter spacetime and the Agmon estimates. We then deduce a logarithmic lower bound for the local energy decay of the Dirac propagator.}

\section{Introduction}

There has been a lot of works concerning the stability of black holes spacetimes in mathematical general relativity. Recently, some significant results have been obtained in this direction. Concerning the Schwarzschild solution, M. Dafermos, G. Holzegel and I. Rodnianski \cite{DaHoRo16} obtained linear stability to gravitational perturbation by means of vector field methods. F. Finster and J. Smoller \cite{FinSmo16} obtained linear stability of the non-extreme Kerr black hole using an integral representation for the solution of the Teukolsky equation. Concerning the full non-linear stability problem, P. Hintz and A. Vasy \cite{HiVa16} obtained the non-linear stability of the Kerr-de Sitter family of black hole by means of a general framework which uses microlocal techniques. \\
In particular, the latest uses a precise understanding of the resonances for their problem. Resonances allow to obtain a precise decay rate for the fields and at which frequency this decay happens. This has been a subject of advanced study in the last decades. In mathematical general relativity, the analysis of resonances, which are also called quasi-normal modes, begins with the work of Bachelot and Motet-Bachelot \cite{BaMBa93} for the Schwarzschild black hole. This study was then pursued by A. S\'{a} Barreto and M. Zworski \cite{SBZ97} for spherically symmetric black holes where resonances are presented as poles of the meromorphic continuation of the resolvent. J-F Bony and D. Häfner \cite{BoHa} gave a full expansion of the solution of the wave equation in terms of resonant states which gives an exponential decay of the local energy for the Schwarzschild-de Sitter solution. R. Melrose, A. S\'{a} Barreto and A. Vasy \cite{MeSBV14} extended this result to more general manifolds and initial data. This was also extended to the case of rotating black hole with positive cosmological constant by S. Dyatlov \cite{Dya11a}, \cite{Dya12} and then by A. Vasy \cite{Va13} who develops more general and flexible techniques. The method developed by S. Dyatlov was also extended to the framework of Dirac operators for the Kerr-Newman-de Sitter black hole by A. Iantchenko \cite{Ia15b} to define resonances . This result generalizes papers by the same author concerning resonances in the de Sitter-Reissner-Nordström spacetime where expansion of the solution in term of resonant states was also given \cite{Ia14}, \cite{Ia15}. \\
Concerning resonances in Anti-de Sitter black holes, C. Warnick \cite{War14} defines resonances in asymptotically Anti-de Sitter black holes. Quasi-normal modes were also defined for Kerr-Anti-de Sitter black holes by O. Gannot \cite{Gan16}. In the case of Schwarzschild-Anti-de Sitter black hole, the same author proves the existence of resonances using the black box method. By means of exponentially accurate quasimodes, he obtains the localization of resonances exponentially close to the real axis \cite{Gan14}. This is linked to a logarithmic local energy decay for the Klein-Gordon fields that was proven by G. Holzegel and J. Smulevici \cite{HolSmu2}, \cite{HolSmu3} in the more general setting of the Kerr-Anti-de Sitter family of black holes. This seems to indicate that the Kerr-Anti-de Sitter spacetime may not be stable, at least for the reflecting boundary conditions. For Schwarzschild-Anti-de Sitter, G. Holzegel and J. Smulevici \cite{Holsmu} proved the stability of this spacetime with respect to spherically symmetric perturbation where the trapping causing the instability is not seen. When looking at dissipative boundary conditions, G. Holzegel, J. Luk, J. Smulevici and C. Warnick \cite{HoLuSmuWar15} prove some estimates that should be useful for proving the non-linear stability of anti-de Sitter space. For higher dimensional Anti-de Sitter spaces, A. Bachelot studied the Cauchy problem for the Klein-Gordon fields in $AdS^{5}$ \cite{Ba11} and investigate the decay of these fields in several settings. The same author studied the propagation of gravitational waves in a Poincaré patch of $AdS^{5}$ \cite{Ba13} where a new Hilbert functional framework is introduced.\\
Although there's a significant amount of results when considering spacetimes with a negative cosmological constant, most of them are concerned with the Klein-Gordon fields. As far as we know, few things are known concerning Dirac fields for spacetimes with a negative cosmological constant. The first result concerning these fields was obtained by A. Bachelot \cite{Bachelot} where he studied the well posedness of the Dirac equation on the Anti-de Sitter space and show the existence of a B-F bound. After having reduced the equation to a spectral problem, this relies on a careful analysis of the states near the boundary. Concerning the Schwarzschild-Anti-de Sitter black hole, we obtained asymptotic completeness for massive Dirac fields in \cite{IR14}. This result shows that the particles behave asymptotically like a transport at speed $1$ in the direction of the black hole.
\bigbreak
In order to have a better understanding of the Dirac fields in the Schwarzschild-Anti-de Sitter spacetime, we are now interested in the local properties of these fields. In particular, this paper is devoted to the proof of a logarithmic lower bound on the local energy decay for the massive Dirac fields in the Schwarzschild-Anti-de Sitter spacetime.\\
More precisely, the massive Dirac equation can be put in the following form:
\begin{equation}
\partial_{t} \phi = i H_{m} \phi,
\end{equation}
where:
\begin{equation}
H_{m} = i\gamma^{0}\gamma^{1} \partial_{x}  + iA\pare{x} \cancel{D}_{\mathbb{S}^{2}} -m \gamma^{0}B\pare{x},
\end{equation}
with the Dirac matrices $\gamma^{i}$. When considering this equation for $x\in ]-\infty,0[$ with reflecting boundary condition, it was proven in a precedent paper \cite{IR14} that this is well-posed. After having given some properties of these fields in the asymptotic regions of our spacetime in \cite{IR14}, we are now interested in the decay of the local energy for these fields. In fact, when looking at the potentials $A$ and $B$, we see that $B$ is confining near $0$ whereas $A$ is having the usual non degenerate maximum due to the photon sphere. Even though it cannot be the same picture for Dirac fields due to the matrices, it is good to keep in mind what is happening for the Klein-Gordon fields where this kind of behavior creates a well where the particles are classically trapped. In this picture, this seems to indicate that the local energy will decay slowly. In fact, for the massive Dirac fields, this decay is not faster than logarithmic. This is the subject of the main result of this paper:
\begin{theo}
For all compact $K \subset ]-\infty,0[$, there exists a constant $C>0$ such that:
\begin{equation*}
\underset{t \to +\infty}{\limsup} \underset{\varphi \in \mathcal{H}, \norme{\varphi}= 1}{\sup} \ln\pare{t}\norme{e^{itH}\varphi}_{L^{2}\pare{K}} \geq C.
\end{equation*}
\end{theo}
To prove this result, we will look at the generalized eigenvalue problem for a semi-classical Dirac operator that we introduce now:
\begin{equation}
H = i\gamma^{0}\gamma^{1} h\partial_{x}  + \gamma^{0}\gamma^{2}A\pare{x} -hm \gamma^{0} B\pare{x}.
\end{equation}
This operator will be considered on the half-line $]-\infty,0[$ with reflecting boundary conditions. Coming back to the picture for Klein-Gordon fields, we could think that the presence of the well will create eigenvalues. However, thanks to the exponential decay of our potentials at $-\infty$ and to tunneling effects, this is not the case. Nevertheless, it is still possible to obtain approximate solutions to the generalized eigenvalue problem for energies close to the bottom of the potential, where the remainder is exponentially decaying in term of the semi-classical parameter. These solutions are also called quasimodes. This is the subject of the following theorem:
\begin{theo}
Let $S>0$ such that $\frac{1}{l^{2}}+S < A^{2}\pare{x_{+}}$. There exist constants $h_{0} >0$, $D >0$, a real number $\pare{E^{+}\pare{h}}^{\frac{1}{2}}>0$ such that $E^{+}\pare{h} < \frac{1}{l^{2}} + S$ and a function $\varphi \in D\pare{H}$ such that $\norme{\varphi}_{L^{2}\pare{]-\infty,0[}} = 1$ satisfying:
\begin{equation*}
\norme{\pare{H - \pare{E^{+}\pare{h}}^{\frac{1}{2}}}\varphi}_{L^{2}\pare{]-\infty,0[}} \leq e^{-\frac{D}{h}},
\end{equation*}
for all $h \in (0,h_{0})$.
\end{theo}
Let us finally notice that this result is in accordance with the one obtained by G. Holzegel and J. Smulevici in \cite{HolSmu3}.
\bigbreak
This paper is organized as follows. \\
In section $2$, we present the Schwarzschild-Anti-de Sitter black hole and the special feature of its geometry. Then we present the Dirac equation in this spacetime.\\
In section $3$, we construct the exponentially accurate quasimodes for our operator following the proof of \cite{Gan14}. \\
More precisely, in subsection $3.1$, we present the operators that we will use all along our analysis. In particular, we obtain estimates for elements of the domain of the square of our Dirac operator. We then restrict this operator to a compact set and obtain some formulas already like in \cite{Gan14}. We prove existence of an eigenvalue close to the local minimum of the potential for an operator $\tilde{P}$ constructed from our Dirac operator by taking a potential having the same asymptotic behavior as our, but being confining.\\
In subsection $3.2$, we obtain Agmon estimates. The presence of the mass creates a confining potential at the Anti-de Sitter infinity. Combined with the usual shape of the potential near the photon sphere, this creates a well where the particle is classically trapped. Nevertheless, at a quantum level, and because of the presence of the Schwarzschild black hole, we can see a tunneling effect where the probability of finding the particle in the classically forbidden region is exponentially decreasing. That is exactly what the Agmon estimates say. \\
In subsection $3.3$, we gives some estimates on the eigenvalue of $\tilde{P}$.\\
We then end this section by using Agmon estimates and a cut-off argument to construct exponentially accurate quasimodes for our Dirac operator which gives the theorem presented earlier.\\
In section $4$, we obtain a logarithmic lower bound for the local energy decay of the Dirac fields in the Schwarzschild-Anti-de Sitter spacetime following the ideas of \cite{HolSmu3} and obtaining this way a similar result.

\begin{center}
\bf{Acknowledgments}
\end{center}

This work was partially supported by the ANR project AARG.

\section{The Dirac equation on the Schwarzschild Anti-de Sitter spacetime}

\indent In this section, we present the Schwarzschild Anti-de Sitter space-time and give the coordinate system that we will work with in the rest of the paper. We quickly study the radial null geodesics and then formulate the Dirac equation as a system of partial differential equations that can be put under a spectral form.

\subsection{The Schwarzschild Anti-de Sitter space-time}

Let $\Lambda <0$. We define $l^{2}= \frac{-3}{\Lambda}$. We denote by $M$ the black hole mass.

In Boyer-Lindquist coordinates, the Schwarzschild-Anti-de Sitter metric is given by:
\begin{equation}
g_{ab}=\pare{1-\frac{2M}{r}+\frac{r^{2}}{l^{2}}} dt^{2} - \pare{1-\frac{2M}{r}+\frac{r^{2}}{l^{2}}}^{-1} dr^{2}- r^{2} \pare{d\theta^{2} + \sin^{2} \theta d\varphi^{2}}
\end{equation}
We define $F(r)= 1- \frac{2M}{r}+ \frac{r^{2}}{l^{2}}$. We can see that $F$ admits two complex conjugate roots and one real root $r=r_{SAdS}$. We deduce that the singularities of the metric are at $r=0$ and $r=r_{SAdS}=p_{+}+p_{-}$ where $p_{\pm}=\pare{Ml^{2} \pm \pare{M^{2}l^{4} + \frac{l^{6}}{27}}^{\frac{1}{2}}}^{\frac{1}{3}}$. (See \cite{Holsmu}) The exterior of the black hole will be the region $r \geq r_{SAdS}$ and our spacetime is then seen as $\R_{t} \times ]r_{SAdS},+\infty[ \times S^{2}$. It is well-known that the metric can be extended for $r \leq r_{SAdS}$ by a coordinate change which gives the maximally extended Schwarzschild-Anti-de Sitter spacetime. In this paper, we are only interested in the exterior region.

In order to have a better understanding of this geometry, we study the outgoing (respectively ingoing) radial null geodesics. Using the form of the metric we can see that along such geodesics, we have:
\begin{equation}
\frac{dt}{dr} = \pm F\pare{r}^{-1}.
\end{equation}
We thus introduce a new coordinate $x$ such that $t-x$ (respectively $t+x$) is constant along outgoing (respectively ingoing) radial null geodesics. In other words:
\begin{equation}
\frac{\mathrm dx}{\mathrm dr}= F(r)^{-1}.
\end{equation}
With this new coordinate, we have:
\begin{flalign}
\hphantom{A}& \lim_{r \to r_{SAdS}} x\pare{r}=- \infty \\
\hphantom{A}& \lim_{r \to \infty} x\pare{r}= 0.
\end{flalign}
This limit proves that, along radial null geodesic, a particle goes to timelike infinity in finite Boyer-Lindquist time (recall that along these geodesic, $t-r_{*}$ and $t+r_{*}$ are constants). This geometric property will lead to several differences compared to other spacetime. In the massless case, we would have to put boundary condition. In fact, even in the massive case, a confining potential appear near the Anti-de Sitter infinity. Combined with the usual bump created by the photon sphere, we expect that the energy will not decay really fast. In fact, for the wave equation, the local energy is decaying only logarithmically as was proven by G. Holzegel and J. Smulevici \cite{HolSmu2}, \cite{HolSmu3}. In this work, we obtain a logarithmic lower bound as in \cite{HolSmu3}.

\subsection{The Dirac equation on this spacetime}

Using the 4-component spinor $
\psi=\begin{pmatrix}
\psi_{1}\\
\psi_{2} \\
\psi_{3} \\
\psi_{4}
\end{pmatrix}$, the Dirac equation in the Schwarzschild-Anti-de Sitter spacetime takes the form:
\begin{equation} \label{opé1}
\left( \partial_{t} + \gamma^{0}\gamma^{1} \pare{F(r)\partial_{r} + \frac{F\pare{r}}{r} + \frac{F'\pare{r}}{4}} +\frac{F(r)^{\frac{1}{2}}}{r} \cancel{D}_{\mathbb{S}^{2}} + im \gamma^{0} F(r)^{\frac{1}{2}}\right) \psi = 0. 
\end{equation}
where $m$ is the mass of the field and $\cancel{D}_{\mathbb{S}^{2}}$ is the Dirac operator on the sphere. In the coordinate system given by $\pare{\theta,\varphi}\in [0;2\pi]\times [0;\pi]$, we obtain: $\cancel{D}_{\mathbb{S}^{2}}= \gamma^{0} \gamma^{2} \pare{ \partial_{\theta} + \frac{1}{2} \cot \theta} + \gamma^{0} \gamma^{3} \frac{1}{\sin \theta} \partial_{\varphi}$ where singularities appear, but we just have to change our chart in this case. We will now work in these coordinates. For more details about how to obtain this form of the equation, we refer to a previous work \cite{IR14}.\\
Recall that the Dirac matrices $\gamma^{\mu}$, $0\leq \mu \leq 3$, unique up to unitary transform, are given by the following relations:
\begin{equation}
\gamma^{0^{*}} = \gamma^{0}; \hspace{3mm} \gamma^{j^{*}} = -\gamma^{j},\hspace{3mm} 1\leq j \leq 3;\hspace{3mm} \gamma^{\mu} \gamma^{\nu} + \gamma^{\nu} \gamma^{\mu} = 2 g^{\mu \nu}\mathbf{1},\hspace{3mm} 0\leq \mu, \nu \leq 3.
\end{equation}
In our representation, the matrices take the form:
\begin{equation} \label{MatDir}
\gamma^{0} = i\begin{pmatrix}
0 & \sigma^{0} \\
-\sigma^{0} & 0
\end{pmatrix}, \hspace{2mm} 
\gamma^{k} = i \begin{pmatrix}
0 & \sigma^{k} \\
\sigma^{k} & 0
\end{pmatrix}, \hspace{2mm} k=1,2,3
\end{equation}
where the Pauli matrices are given by:
\begin{equation}
\sigma^{0}=\begin{pmatrix}
1&0 \\
0 & 1
\end{pmatrix}, \hspace{1mm} 
\sigma^{1}= \begin{pmatrix}
1 & 0 \\
0 & -1
\end{pmatrix}, \hspace{1mm}
\sigma^{2}= \begin{pmatrix}
0 & 1 \\
1 & 0
\end{pmatrix}, \hspace{1mm}
\sigma^{3}= \begin{pmatrix}
0 & -i \\
i & 0
\end{pmatrix}.
\end{equation}
We thus obtain:
\begin{equation}
\gamma^{0}\gamma^{1} = \begin{pmatrix} 
-\sigma^{1}& 0 \\
0& \sigma^{1} \\
\end{pmatrix}; \hspace{2mm} \gamma^{0}\gamma^{2} = \begin{pmatrix}
-\sigma^{2}& 0 \\
0 &\sigma^{2}
\end{pmatrix};\hspace{2mm} \gamma^{0} \gamma^{3} = \begin{pmatrix}
-\sigma^{3}& 0 \\
0& \sigma^{3}
\end{pmatrix}.\label{refGamma0,1,2,3}
\end{equation}
We introduce the matrix:
\begin{equation}
\gamma^{5} = -i\gamma^{0}\gamma^{1}\gamma^{2} \gamma^{3}
\end{equation}
which satisfies the relations:
\begin{equation}
\gamma^{5}\gamma^{\mu} + \gamma^{\mu}\gamma^{5} = 0,\hspace{5mm} 0\leq \mu \leq 3. \label{relGamma5}
\end{equation}
We make the change of spinor $\phi(t,x,\theta,\varphi) = r F(r)^{\frac{1}{4}} \psi (t,r,\theta,\varphi)$ and obtain the following equation:
\begin{equation}
\partial_{t} \phi = i \left(i\gamma^{0}\gamma^{1} \partial_{x}  + i\frac{F(r)^{\frac{1}{2}}}{r} \cancel{D}_{\mathbb{S}^{2}} -m \gamma^{0} F(r)^{\frac{1}{2}} \right) \phi.
\end{equation}
We set:
\begin{equation}
H_{m} = i\gamma^{0}\gamma^{1} \partial_{x}  + i\frac{F(r)^{\frac{1}{2}}}{r} \cancel{D}_{\mathbb{S}^{2}} -m \gamma^{0} F(r)^{\frac{1}{2}}.
\end{equation}
We introduce the Hilbert space:
\begin{equation}
\mathcal{H} := \left [L^{2}\pare{\left]-\infty,0\right[_{x} \times S^{2}_{\omega}, dx d\omega} \right ]^{4}
\end{equation}
Recall that $r$ is now a function of $x$. Using the spinoidal spherical harmonics (see \cite{IR14} for the details), we are able to diagonalize the Dirac operator on the sphere and we obtain the following operator:
\begin{equation*}
H_{m} = i\gamma^{0}\gamma^{1} \partial_{x}  + \gamma^{0}\gamma^{2}\frac{F(r)^{\frac{1}{2}}}{r} \pare{s+\frac{1}{2}} -m \gamma^{0} F(r)^{\frac{1}{2}}.
\end{equation*}
The corresponding Hilbert space is now:
\begin{equation}
\mathcal{H}_{s,n} := \left [L^{2}\pare{\left]-\infty,0\right[_{x}} \right]^{4}\otimes Y_{s,n}
\end{equation}
where $Y_{s,n}$ span the corresponding spinoidal spherical harmonic for harmonic $s$ and $n$ fixed.\\
Next, we introduce a semi-classical parameter $h = \frac{1}{s+\frac{1}{2}}$. In the remaining of this article, we will deal with the following operator:
\begin{equation}
H = hH_{m} 
= i\gamma^{0}\gamma^{1} h\partial_{x}  + \gamma^{0}\gamma^{2}\frac{F(r)^{\frac{1}{2}}}{r} -hm \gamma^{0} F(r)^{\frac{1}{2}}
\end{equation}
where we have omitted to put the indice $m$ for convenience. It was proven in \cite{IR14}, that this operator is self-adjoint for all positive masses when equipped with the appropriate domain. In the sequel, we will be interested in the case where the mass is sufficiently large. The domain of self-adjointness of our operator is then its natural domain, namely:
\begin{equation*}
D\pare{H} = \{\varphi \in \mathcal{H}_{s,n} \lvert H\varphi \in \mathcal{H}_{s,n} \}.
\end{equation*}

\section{Quasimodes}

In this section, we construct exponentially accurate quasimodes for our Dirac operator. We will follow the same path as the one described in \cite{Gan14}.

\subsection{Description of the operators}

We introduce new comparison operator in the same spirit as in \cite{Gan14}. First, recall that
\begin{equation*}
H = i\gamma^{0}\gamma^{1}h \partial_{x} + \gamma^{0}\gamma^{2} A\pare{x} - h m \gamma^{0} B\pare{x}.
\end{equation*}
Here, $A$ and $B$ behave at $0$ like:
\begin{equation*}
 A\pare{x} = \frac{1}{l} + \frac{x^{2}}{2l^{3}} + o\pare{x^{2}}, \hspace{5mm} B\pare{x} = -\frac{l}{x} - \frac{x}{6l} + o\pare{x}. 
\end{equation*}
Thus:
\begin{equation*}
 A^{2}\pare{x}  = \frac{1}{l^{2}} + \frac{x^{2}}{l^{4}} + o\pare{x^{2}}, \hspace{5mm}
B^{2}\pare{x}  = \frac{l^{2}}{x^{2}} + \frac{1}{3l} + o\pare{1},
\end{equation*}
and:
\begin{equation*}
 A'\pare{x}  = \frac{x}{l^{3}} + o\pare{x}, \hspace{5mm}
B'\pare{x} = \frac{l}{x^{2}} - \frac{1}{6l} + o\pare{1}.
\end{equation*}
We will consider the case $ml>1$ for which the domain of our operator $H$ is:
\begin{equation*}
D\pare{H} = \{ \varphi \in L^{2} \lvert H \varphi \in L^{2} \}.
\end{equation*}
Furthermore, we have:
\begin{equation*}
\pare{\gamma^{0}}^{2} = Id, \hspace{1mm} \pare{\gamma^{j}}^{2} = -Id, \hspace{1mm} \gamma^{\mu}\gamma^{\nu} = - \gamma^{\nu}\gamma^{\mu},
\end{equation*}
pour $1\leq j \leq 3$ et $0 \leq \mu, \nu \leq 3$. The square of $H$ is:
\begin{equation*}
P = H^{2}  = -h^{2}\partial_{x}^{2} +V\pare{x},
\end{equation*}
where:
\begin{equation} \label{ExprV}
V\pare{x} = A^{2}\pare{x} + h^{2}m^{2}B^{2}\pare{x}- ih \gamma^{1}\gamma^{2} A'\pare{x} + ih^{2}m \gamma^{1} B'\pare{x}.
\end{equation}
We equip $P$ with the domain:
\begin{equation*}
D\pare{P} = \{\varphi \in L^{2} \lvert H\varphi \in L^{2}, P\varphi \in L^{2} \}
\end{equation*}
on which it is well-defined and self-adjoint. We have $\gamma^{0} \gamma^{1} = diag\pare{-1,1,1,-1}$ and:
\begin{flalign*}
 \gamma^{1}\gamma^{2} = \begin{pmatrix}
0 & 1 & 0 & 0 \\
-1 & 0 & 0 & 0 \\
0 & 0 & 0 & 1 \\
0 & 0 & -1 & 0
\end{pmatrix}, & \gamma^{1} = i \begin{pmatrix}
0 & 0 & 1 & 0 \\
0 & 0 & 0 & -1 \\
1 & 0 & 0 & 0 \\
0 & -1 & 0 & 0
\end{pmatrix}, \\
\gamma^{0}\gamma^{2} = \begin{pmatrix}
0 & -1 & 0 & 0 \\
-1 & 0 & 0 & 0 \\
0 & 0 & 0 & 1 \\
0 & 0 & 1 & 0
\end{pmatrix}, & \gamma^{0} = \begin{pmatrix}
0 & 0 & i & 0 \\
0 & 0 & 0 & i \\
-i & 0 & 0 & 0 \\
0 & -i & 0 & 0 
\end{pmatrix}.
\end{flalign*}
First, we study elements of the domain of $P$. 

\begin{prop}
Let $\varphi \in D\pare{P}$. Near $0$, we obtain the following behavior:
\begin{enumerate}
\item[-] If $2ml \neq 3$, $\norme{\varphi\pare{x}}_{\C^{4}} = O\pare{\pare{-x}^{\min \pare{\frac{3}{2},ml}}}$.
\item[-] If $2ml = 3$, $\norme{\varphi\pare{x}}_{\C^{4}}= O\pare{\max \pare{\pare{-x}^{\frac{3}{2}}, \pare{-x}^{ml}\pare{-\ln\pare{-x}}}}$.
\end{enumerate}
We also obtain the following estimate for $\partial_{x}\varphi$:
\begin{flalign*}
\norme{\partial_{x} \varphi\pare{x}} & = \begin{cases}
O\pare{\max\pare{\pare{-x}^{\frac{1}{2}}, \pare{-x}^{\min \pare{\frac{3}{2},ml}-1}}},  \hspace{2mm} \text{if} \hspace{2mm} 2ml \neq 3, \\
O\pare{\max\pare{\pare{-x}^{\frac{1}{2}}, \pare{-x}^{ml-1}\pare{-\ln\pare{-x}}}},  \hspace{2mm} \text{if} \hspace{2mm} 2ml = 3.
\end{cases}
\end{flalign*}
\end{prop}

\begin{proof}
Let $\varphi \in D\pare{P}$. Since $H^{2} \varphi \in L^{2}$ and using a result from \cite{IR14} saying that, for $\psi \in D\pare{H}$ and $m > \frac{1}{2l}$, we have
\begin{equation}
|| \psi(x,.)||_{L^2(S^2)} = O \left( \sqrt{-x} \right), \hspace{2mm} x \to 0,\label{705}
\end{equation}
we deduce that:
\begin{equation*}
\norme{H\varphi\pare{x}} = O\pare{\pare{-x}^{\frac{1}{2}}}
\end{equation*}
near $0$. Thus, we can write:
\begin{equation*}
H \varphi \pare{x} = f\pare{x}
\end{equation*}
with $f\pare{x} = O\pare{\pare{-x}^{\frac{1}{2}}}$. We obtain the following equations:
\begin{flalign*}
\begin{cases}
-i h \partial_{x} \varphi_{1} - A\pare{x} \varphi_{2} - ihm B\pare{x} \varphi_{3} = f_{1} \\
ih\partial_{x} \varphi_{2} - A\pare{x} \varphi_{1} - ihm B\pare{x} \varphi_{4} = f_{2} \\
ih \partial_{x} \varphi_{3} + A\pare{x} \varphi_{4} +ihm B\pare{x} \varphi_{1} = f_{3} \\
-ih \partial_{x} \varphi_{4} + A\pare{x} \varphi_{3} + ihm B\pare{x} \varphi_{2} = f_{4}.
\end{cases}
\end{flalign*}
This gives the two following systems:
\begin{flalign}
\hphantom{A} & \begin{cases}
-ih\partial_{x} \pare{\varphi_{1} + \varphi_{3}}-ihm B\pare{x} \pare{\varphi_{1} + \varphi_{3}} = g_{1} \\
-ih\partial_{x} \pare{\varphi_{1} - \varphi_{3}} + ihm B\pare{x} \pare{\varphi_{1} - \varphi_{3}} = g_{2} \label{H21}
\end{cases} \\
& \begin{cases}
ih \partial_{x} \pare{\varphi_{2} + \varphi_{4}}- ihm B\pare{x} \pare{\varphi_{2} + \varphi_{4}} = g_{3} \\
ih \partial_{x} \pare{\varphi_{2} - \varphi_{4}} + ihm B\pare{x} \pare{\varphi_{2} - \varphi_{4}} = g_{4},
\end{cases}
\end{flalign}
where:
\begin{flalign*}
g_{1}\pare{x} = A\pare{x} \pare{\varphi_{2}+\varphi_{4}} + f_{1} - f_{3}, 
& \hspace{5mm} g_{2}\pare{x} = A\pare{x} \pare{\varphi_{2} - \varphi_{4}} + f_{1} + f_{3},\\
g_{3}\pare{x} = A\pare{x} \pare{\varphi_{1} + \varphi_{3}} + f_{2} - f_{4},
& \hspace{5mm} g_{4}\pare{x} = A\pare{x} \pare{\varphi_{1} - \varphi_{3}} + f_{2} + f_{4}.
\end{flalign*}
Since $A$ is bounded near $0$ and $H\varphi \in L^{2}$, we have $\norme{\varphi\pare{x}} = O\pare{\pare{-x}^{\frac{1}{2}}}$. Thus, $\norme{g\pare{x}} = O\pare{\pare{-x}^{\frac{1}{2}}}$. We first study equations $\eqref{H21}$.\\
\underline{First equation}\\
Here, we study the equation:
\begin{equation*}
\partial_{x} \pare{\varphi_{1} + \varphi_{3}} + m B\pare{x} \pare{\varphi_{1} + \varphi_{3}} = \frac{i}{h} g_{1}.
\end{equation*}
We write $\varphi_{1,3}^{+}= \varphi_{1} + \varphi_{3}$. The homogeneous equation has a solution of the form:
\begin{equation*}
\varphi_{1,3}^{+,EH}\pare{x} = C e^{-\int_{-\epsilon}^{x} m B\pare{y} dy }
\end{equation*}
where $\epsilon > 0$ and $C$ is a constant. With a particular solution of the form:
\begin{equation*}
\varphi_{1,3}^{+,p}\pare{x} = C\pare{x} e^{-\int_{-\epsilon}^{x} m B\pare{y} dy },
\end{equation*}
we have:
\begin{equation*}
C\pare{x} = \int_{-\epsilon}^{x} \frac{i}{h} g\pare{y} e^{\int_{-\epsilon}^{y} m B\pare{z} dz } dy.
\end{equation*}
Since, near $0$, we have:
\begin{equation*}
B\pare{x} = -\frac{l}{x} - \frac{x}{6l} + o\pare{x},
\end{equation*}
so that:
\begin{equation*}
e^{\int_{-\epsilon}^{x} m B\pare{z} dz} = \pare{-x}^{-ml}O\pare{1},
\end{equation*}
we obtain:
\begin{flalign*}
\abs{C\pare{x}} & \leq C \int_{-\epsilon}^{x} \frac{1}{h} \abs{\pare{-y}^{\frac{1}{2}}} \abs{\pare{-y}^{-ml}} dy \\
& = \begin{cases}
\frac{C}{h} \left [ \frac{\pare{-y}^{\frac{3}{2} - ml}}{\frac{3}{2} - ml} \right ]_{-\epsilon}^{x}, \hspace{2mm} \text{if} \hspace{2mm} ml \neq \frac{3}{2} \\
\frac{C}{h} \left [ - \ln\pare{-y} \right ]_{-\epsilon}^{x}, \hspace{2mm} \text{if} \hspace{2mm} ml= \frac{3}{2}.
\end{cases}
\end{flalign*}
With $e^{-\int_{-\epsilon}^{x} m B\pare{y} dy} = \pare{-x}^{ml} O\pare{1}$, we obtain:
\begin{flalign*}
\begin{cases}
\text{If} \hspace{2mm} 2ml \neq 3, \hspace{2mm} \varphi_{1,3}^{+,p}\pare{x} = O\pare{\pare{-x}^{\frac{3}{2}}},\\
\text{If} \hspace{2mm} 2ml = 3, \hspace{2mm} \varphi_{1,3}^{+,p}\pare{x} = O \pare{\pare{-x}^{ml}\pare{-\ln\pare{-x}}}.
\end{cases}
\end{flalign*}
Moreover, $\varphi_{1,3}^{+,EH}\pare{x} = O\pare{\pare{-x}^{ml}}$, which gives:
\begin{enumerate}
\item[-] If $2ml \neq 3$, $\varphi_{1,3}^{+}\pare{x} = O\pare{\pare{-x}^{\min \pare{\frac{3}{2}, ml}}}$.
\item[-] If $2ml = 3$, $\varphi_{1,3}^{+}\pare{x} = O\pare{\pare{-x}^{ml}\pare{-\ln\pare{-x}}}$.
\end{enumerate}
\underline{Second equation}\\
We write $\varphi_{1,3}^{-} = \varphi_{1} - \varphi_{3}$. We have the equation:
\begin{equation*}
\partial_{x} \varphi_{1,3}^{-} \pare{x} - m B\pare{x} \varphi_{1,3}^{-}\pare{x} = \frac{i}{h}g_{2}\pare{x}.
\end{equation*}
As before, a solution of the homogeneous equation is:
\begin{equation*}
\varphi_{1,3}^{-,EH}\pare{x} = C e^{\int_{-\epsilon}^{x} m B\pare{y} dy}
\end{equation*}
with a constant $C$. This solution behaves like $\varphi_{1,3}^{-,EH}\pare{x} = C\pare{-x}^{-ml}O\pare{1}$. A particular solution is:
\begin{equation*}
\varphi_{1,3}^{-,p}\pare{x} = e^{\int_{-\epsilon}^{x} m B\pare{y} dy} \int_{0}^{x} \frac{i}{h} g_{2}\pare{y} e^{-\int_{-\epsilon}^{y} m B\pare{z} dz} dy.
\end{equation*}
Here, we have $e^{-\int_{-\epsilon}^{y} m B\pare{z} dz} = \pare{-y}^{ml}O\pare{1}$ which gives:
\begin{equation*}
\abs{\varphi_{1,3}^{-,p}\pare{x}}  \leq C_{1} \pare{-x}^{-ml} \int_{0}^{x} \frac{1}{h} \pare{-y}^{\frac{1}{2}+ ml} dy \leq C_{1} \pare{-x}^{\frac{3}{2}}
\end{equation*}
Since $\varphi \in D\pare{P}$, we have $H\varphi \in L^{2}$ and $\norme{\varphi\pare{x}} = O\pare{\pare{-x}^{\frac{1}{2}}}$ using a result in \cite{IR14}. It implies that $C=0$. We then have:
\begin{equation*}
\varphi_{1,3}^{-}\pare{x} = O\pare{\pare{-x}^{\frac{3}{2}}}.
\end{equation*}
This gives the estimates for $\varphi_{1}$ and $\varphi_{3}$:
\begin{enumerate}
\item[-] If $2ml \neq 3$, $\varphi_{1}\pare{x} = O\pare{\pare{-x}^{\min \pare{\frac{3}{2},ml}}}$ and $\varphi_{3}\pare{x} = O\pare{\pare{-x}^{\min \pare{\frac{3}{2},ml}}}$.
\item[-] If $2ml = 3$, $\varphi_{1}\pare{x} = O\pare{\max \pare{\pare{-x}^{\frac{3}{2}}, \pare{-x}^{ml}\pare{-\ln\pare{-x}}}}$ and the same estimate for $\varphi_{3}$.
\end{enumerate}
We can do the same with $\varphi_{2}$ et $\varphi_{4}$. Using that, if $\varphi \in D\pare{P}$, we have $H \varphi \in D\pare{H}$ and $H\varphi\pare{x} = O\pare{\pare{-x}^{\frac{1}{2}}}$ we deduce that:
\begin{equation*}
i \gamma^{0} \gamma^{1} \partial_{x} \varphi\pare{x} = - \gamma^{0}\gamma^{2} A\pare{x} \varphi\pare{x} + hm \gamma^{0} B\pare{x} \varphi\pare{x} + O\pare{\pare{-x}^{\frac{1}{2}}}.
\end{equation*}
With the previous estimates, we obtain:
\begin{flalign*}
- \gamma^{0}\gamma^{2} A\pare{x} \varphi\pare{x} & = O\pare{\pare{-x}^{\frac{1}{2}}} \\
hm\gamma^{0} B\pare{x} \varphi\pare{x} & = \begin{cases} O \pare{\pare{-x}^{\min \pare{\frac{3}{2},ml}-1}}, \hspace{2mm} \text{if} \hspace{2mm} 2ml \neq 3, \\
O\pare{\max\pare{\pare{-x}^{\frac{1}{2}}, \pare{-x}^{ml-1}\pare{-\ln\pare{-x}}}},  \hspace{2mm} \text{if} \hspace{2mm} 2ml = 3.
\end{cases}
\end{flalign*}
So:
\begin{flalign*}
\norme{\partial_{x} \varphi\pare{x}} & = \begin{cases}
O\pare{\max\pare{\pare{-x}^{\frac{1}{2}}, \pare{-x}^{\min \pare{\frac{3}{2},ml}-1}}},  \hspace{2mm} \text{if} \hspace{2mm} 2ml \neq 3, \\
O\pare{\max\pare{\pare{-x}^{\frac{1}{2}}, \pare{-x}^{ml-1}\pare{-\ln\pare{-x}}}},  \hspace{2mm} \text{if} \hspace{2mm} 2ml = 3.
\end{cases}
\end{flalign*}
\end{proof}
\begin{rem}
In every case, we have:
\begin{equation*}
\prods{\partial_{x}\varphi\pare{x}}{\varphi\pare{x}}_{\C^{4}} \underset{x \to 0}{\to} 0.
\end{equation*}
\end{rem}
We now consider $P$ on $J= ]-c,0[$ where $c$ is to be determined. We denote $P_{J}$ this operator and equip it with the domain:
\begin{equation*} \label{DomPComp}
D\pare{P_{J}} = \{\varphi \in L^{2}\pare{J} \lvert H\varphi \in L^{2}\pare{J}, P_{J}\varphi \in L^{2}\pare{J}, \varphi\pare{-c} = 0 \}.
\end{equation*} 
We have the same lemmas as in \cite{Gan14}:
\begin{lem}
Let $\varphi \in D\pare{P_{J}}$ and $Y\pare{x}$ such that $Y\pare{x} = O\pare{x}$, $Y'\pare{x} = O\pare{1}$ at $0$.
Then:
\begin{equation*}
\norme{-ih\partial_{x} \varphi - i FY \varphi}^{2} = \prods{-h^{2} \partial_{x}^{2} \varphi}{\varphi}+ \prods{F^{2}Y^{2}\varphi-hF\partial_{r}\pare{FY}\varphi}{\varphi}
\end{equation*}
where $F\pare{r} = 1 - \frac{2M}{r} + \frac{r^{2}}{l^{2}}$.
\end{lem}

\begin{proof}
For $\varphi \in D\pare{P_{J}}$, we have:
\begin{flalign*}
\norme{-ih\partial_{x} \varphi - i FY \varphi}^{2} & = \prods{-ih\partial_{x} \varphi - i FY \varphi}{-ih\partial_{x} \varphi - i FY \varphi} \\
& = \prods{-h^{2} \partial_{x}^{2} \varphi}{\varphi} - h \prods{\partial_{x}\pare{FY} \varphi}{\varphi} + \prods{F^{2}Y^{2}\varphi}{\varphi},
\end{flalign*}
where the first integration by part is justified by $\prods{\partial_{x} \varphi\pare{x}}{\varphi\pare{x}}_{\C^{4}} \underset{x \to 0}{\to} 0$. In a neighborhood of $0$, we have $F\pare{r\pare{x}} = O\pare{\pare{-x}^{-2}}$, $Y\pare{x} = O\pare{-x}$ et $\varphi\pare{x} = o\pare{\pare{-x}^{\frac{1}{2}}}$ which give $\prods{FY\varphi\pare{x}}{\varphi\pare{x}}_{\C^{4}\lvert x=0} = 0$ and justify the other integration by part. Since $\partial_{x} = \frac{\partial r}{\partial x} \partial_{r}= F\pare{r} \partial_{r}$ we obtain:
\begin{equation*}
\norme{-ih\partial_{x} \varphi - i FY \varphi}^{2} = \prods{-h^{2} \partial_{x}^{2} \varphi}{\varphi}+ \prods{F^{2}Y^{2}\varphi-hF\partial_{r}\pare{FY}\varphi}{\varphi}.
\end{equation*}
\end{proof}

\begin{lem}
Let $r_{+}$ such that, for $r \geq r_{+}$, we have $\frac{F\pare{r}}{4l^{2}} - \frac{F'\pare{r}^{2}}{16} \geq 0$ where $F\pare{r} = 1 - \frac{2M}{r} + \frac{r^{2}}{l^{2}}$. Denote by $x_{+} = x\pare{r_{+}}$ and $P^{+} = P_{]x_{+},0[}$. Let $\varphi \in D\pare{P^{+}}$, then:
\begin{equation*}
\prods{\pare{-h^{2}\partial_{x}^{2} + h^{2}m^{2} B^{2}\pare{x} + i h^{2}m \gamma^{1} B'\pare{x}}\varphi}{\varphi} \geq 0.
\end{equation*}
\end{lem}

\begin{proof}
We have $B^{2}\pare{x\pare{r}} = F\pare{r}$ and:
\begin{equation*}
\partial_{x}B\pare{x\pare{r}} = \partial_{x} F\pare{r}^{\frac{1}{2}} = F\pare{r} \partial_{r} F\pare{r}^{\frac{1}{2}} = \frac{1}{2} F'\pare{r} F\pare{r}^{\frac{1}{2}}.
\end{equation*}
Moreover:
\begin{equation*}
\frac{1}{2} F'\pare{r} F\pare{r}^{\frac{1}{2}}  \leq m F\pare{r} + \frac{F'\pare{r}^{2}}{16m}.
\end{equation*}
We deduce that:
\begin{flalign*}
m^{2} F\pare{r} + im \gamma^{1} \frac{1}{2} F'\pare{r} F\pare{r}^{\frac{1}{2}} & \geq m^{2} F\pare{r} - \frac{m}{2} F'\pare{r}F\pare{r}^{\frac{1}{2}} \\
& \geq - \frac{F'\pare{r}^{2}}{16} = - \frac{F'\pare{r}^{2}}{16} + \frac{F\pare{r}}{4l^{2}} - \frac{F\pare{r}}{4l^{2}},
\end{flalign*}
in the sense of quadratic forms. Since $F'\pare{r} = 2\pare{\frac{M}{r^{2}} + \frac{r}{l^{2}}}$ and $\frac{F\pare{r}}{4l^{2}} = \frac{1}{4l^{2}} - \frac{M}{2rl^{2}} + \frac{r^{2}}{4l^{2}}$, we obtain:
\begin{equation*}
\frac{F\pare{r}}{4l^{2}} - \frac{F'\pare{r}^{2}}{16} = \frac{1}{4l^{2}} - \frac{M}{rl^{2}} - \frac{M^{2}}{4r^{4}}.
\end{equation*}
Thus $\frac{F\pare{r}}{4l^{2}} - \frac{F'\pare{r}^{2}}{16} \geq 0$ for $r$ sufficiently large. Let $r_{+}$ such that, for $r \geq r_{+}$, we have $\frac{F\pare{r}}{4l^{2}} - \frac{F'\pare{r}^{2}}{16} \geq 0$.
Let $x_{+}= x\pare{r_{+}}$. On $[x_{+},0[$, we have:
\begin{equation*}
h^{2}m^{2}B^{2}\pare{x} + i h^{2}m\gamma^{1} \partial_{x}B\pare{x} \geq - h^{2}\frac{F\pare{r}}{4l^{2}}.
\end{equation*}
Consider $Y\pare{x\pare{r}} = h \frac{r-r_{SAdS}}{2l^{2}} F^{-1}\pare{r}$. Since $r = O\pare{\pare{-x}^{-1}}$ and $F\pare{r} = O\pare{\pare{-x}^{-2}}$, we have $Y\pare{x} = O\pare{x}$. Moreover:
\begin{equation*}
\partial_{x}Y\pare{x}  = F\pare{r} \partial_{r} \pare{h \frac{r-r_{SAdS}}{2l^{2}} F^{-1}\pare{r}} = \frac{h}{2l^{2}} - \frac{h}{2l^{2}}\pare{r-r_{SAdS}}\frac{F'\pare{r}}{F\pare{r}}.
\end{equation*}
Since $F'\pare{r\pare{x}} = O\pare{\pare{-x}^{-1}}$, we deduce that $\partial{x}Y\pare{x} = O\pare{1}$.\\
Furthermore, as in \cite{Gan14}, we have:
\begin{equation*}
\partial_{r}\pare{FY} = \frac{h}{2l^{2}}
\end{equation*}
and:
\begin{equation*}
F^{2}Y^{2} - h F\partial_{r}\pare{FY}  = h^{2}\frac{r^{2}- 2rr_{SAdS}+r^{2}_{SAdS}}{4l^{4}} - h^{2} \frac{F\pare{r}}{2l^{2}}  \leq - h^{2} \frac{F\pare{r}}{4l^{2}}.
\end{equation*}
We can now conclude:
\begin{flalign*}
\prods{\pare{-h^{2}\partial_{x}^{2} + h^{2}m^{2} B^{2}\pare{x} + i h^{2}m \gamma^{1} \partial_{x}B\pare{x}} \varphi}{\varphi} & \geq \prods{\pare{-h^{2}\partial_{x}^{2}-h^{2} \frac{F}{4l^{2}}} \varphi}{\varphi} \\
& \geq \prods{\pare{-h^{2}\partial_{x}^{2} + F^{2}Y^{2} - h F\partial_{r}\pare{FY}} \varphi}{\varphi} \\
& = \norme{-ih\partial_{x} \varphi - i FY \varphi}^{2} \geq 0.
\end{flalign*}
\end{proof}
Recall now the asymptotic behavior at $0$ of the potentials:
\begin{equation*}
A\pare{x} = \frac{1}{l} + \frac{x^{2}}{2l^{3}} + o\pare{x^{2}}, \hspace{5mm}
B\pare{x} = -\frac{l}{x} - \frac{x}{6l} + o\pare{x},
\end{equation*}
and:
\begin{flalign*}
A^{2}\pare{x}  = \frac{1}{l^{2}} + \frac{x^{2}}{l^{4}} + o\pare{x^{2}}, & \hspace{5mm}
B^{2}\pare{x}  = \frac{l^{2}}{x^{2}} + \frac{1}{3l} + o\pare{1},\\
A'\pare{x} = \frac{x}{l^{3}} + o\pare{x},& \hspace{5mm}
B'\pare{x} = \frac{l}{x^{2}} - \frac{1}{6l} + o\pare{1}.
\end{flalign*}
We introduce the operator:
\begin{equation*}
\tilde{H} = i\gamma^{0}\gamma^{1} h\partial_{x} + \gamma^{0}\gamma^{2} \pare{\frac{1}{l^{2}} + \frac{x^{2}}{l^{4}}}^{\frac{1}{2}} + hm \gamma^{0} \frac{l}{x},
\end{equation*}
and its square:
\begin{equation*}
\tilde{P}= \tilde{H}^{2} = -h^{2}\partial_{x}^{2} + \pare{\frac{1}{l^{2}} + \frac{x^{2}}{l^{4}}} - ih \gamma^{1}\gamma^{2} \frac{x}{l^{4}}\pare{\frac{1}{l^{2}} + \frac{x^{2}}{l^{4}}}^{-\frac{1}{2}} + h^{2}\pare{\frac{m^{2}l^{2}}{x^{2}} + i\gamma^{1} \frac{ml}{x^{2}}}
\end{equation*}
on $]-\infty,0[$ with domain:
\begin{equation*}
D\pare{\tilde{P}} = \{\varphi\in L^{2} \lvert \tilde{H}\varphi \in L^{2}, \tilde{P} \varphi \in L^{2} \}.
\end{equation*}
The operator $\tilde{H}$ has compact resolvent since its potential is confining. Moreover, we have:
\begin{equation*}
\pare{\tilde{P} - E^{2}}^{-1} = \pare{\tilde{H}+E}^{-1} \pare{\tilde{H} - E}^{-1},
\end{equation*}
which proves that $\tilde{P}$ has compact resolvent and that $E^{2}$ is an eigenvalue of $\tilde{P}$ if and only if $\pm E$ is an eigenvalue of $\tilde{H}$.

\begin{prop} \label{VpPtilde}
Let $\alpha_{1} = \frac{1+\sqrt{1+4ml\pare{ml+1}}}{2}$, $ E_{0}\pare{h}:= \frac{1}{l^{2}} - \frac{h}{2}$ and \\ $E_{2}\pare{h} = \frac{1}{l^{2}} + \pare{2\alpha_{1}+1}\pare{\frac{1}{l^{4}} + \frac{h}{2l^{6}}}^{\frac{1}{2}} h$. Then $\tilde{P}$ admit an eigenvalue $\tilde{E}\pare{h}$ with $E_{0}\pare{h} \leq \tilde{E}\pare{h} \leq E_{2}\pare{h}+\frac{h}{2}$.
\end{prop}

\begin{proof}
Using the fact that $\prods{\partial_{x}\varphi\pare{x}}{\varphi\pare{x}}_{\C^{4}} \underset{x \to 0}{\to} 0$ and that:
\begin{equation*}
\frac{m^{2}l^{2}}{x^{2}} + i\gamma^{1} \frac{ml}{x^{2}} \geq \frac{ml\pare{ml-1}}{x^{2}} \geq 0,
\end{equation*}
since $ml>1$, we have:
\begin{equation*}
-h^{2}\partial_{x}^{2} + h^{2}\pare{ \frac{m^{2}l^{2}}{x^{2}} + i\gamma^{1} \frac{ml}{x^{2}}} \geq 0.
\end{equation*}
Moreover, we have:
\begin{flalign}
\abs{\prods{-ih \gamma^{1}\gamma^{2} \frac{x}{l^{4}} \pare{\frac{1}{l^{2}} + \frac{x^{2}}{l^{4}}}^{-\frac{1}{2}} \varphi}{\varphi}} & \leq  \frac{h}{2}\prods{\frac{x^{2}}{l^{8}} \pare{\frac{1}{l^{2}} + \frac{x^{2}}{l^{4}}}^{-1} \varphi}{\varphi} + \frac{h}{2} \norme{\varphi}^{2} \notag \\
&\leq \frac{h}{2} \prods{\frac{x^{2}}{l^{6}} \varphi}{\varphi} + \frac{h}{2} \norme{\varphi}^{2} \label{*}.
\end{flalign}
This implies:
\begin{equation*}
\frac{1}{l^{2}} + \frac{x^{2}}{l^{4}}-ih \gamma^{1}\gamma^{2} \frac{x}{l^{4}} \pare{\frac{1}{l^{2}} + \frac{x^{2}}{l^{4}}}^{-\frac{1}{2}} \geq \frac{1}{l^{2}} + \frac{2l^{2}x^{2} - hx^{2}}{2l^{6}}-\frac{h}{2}
\geq E_{0}\pare{h},
\end{equation*}
for $h>0$ sufficiently small. Thus:
\begin{equation*}
\prods{\tilde{P}\varphi}{\varphi} \geq E_{0}\pare{h} \norme{\varphi}_{L^{2}}^{2}.
\end{equation*}
Using \eqref{*}, we have:
\begin{equation*}
\prods{\tilde{P}\varphi}{\varphi} \leq \prods{\pare{-h^{2}\partial_{x}^{2} + \frac{1}{l^{2}} + \pare{\frac{1}{l^{4}} + \frac{h}{2l^{6}}}x^{2} + h^{2}\pare{\frac{m^{2}l^{2}}{x^{2}} + i\gamma^{1} \frac{ml}{x^{2}}}} \varphi}{\varphi} + \frac{h}{2} \norme{\varphi}^{2}.
\end{equation*}
We then diagonalize $\gamma^{1}$:
\begin{equation*}
\gamma^{1} = K D K^{-1}
\end{equation*}
where 
\begin{equation*}
K= \frac{1}{\sqrt{2}} \begin{pmatrix}
1 & 0 & 1 & 0 \\
0 & 1 & 0 & 1 \\
1 & 0 & -1 & 0 \\
0 & 1 & 0 & -1
\end{pmatrix}, \hspace{2mm}
D= \text{diag}\pare{i,-i,-i,i}, \hspace{2mm} K^{-1} = K.
\end{equation*}
We can then write:
\begin{flalign*}
\hphantom{A} & \prods{\pare{-h^{2}\partial_{x}^{2} + \frac{1}{l^{2}} + \pare{\frac{1}{l^{4}} + \frac{h}{2l^{6}}}x^{2} + h^{2}\pare{\frac{m^{2}l^{2}}{x^{2}} + i\gamma^{1} \frac{ml}{x^{2}}}} \varphi}{\varphi} \\
& = \left\langle \left (-h^{2}\partial_{x}^{2} + \frac{1}{l^{2}} + \pare{\frac{1}{l^{4}} + \frac{h}{2l^{6}}}x^{2} \right . \right .\\
& \quad \left . \left . + h^{2}\pare{\frac{m^{2}l^{2}}{x^{2}} + \begin{pmatrix}
-1 & 0 & 0 & 0 \\
0 & 1 & 0 & 0 \\
0 & 0 & 1 & 0 \\
0 & 0 & 0 & -1
\end{pmatrix} \frac{ml}{x^{2}}}\right ) K^{-1}\varphi ,K^{-1}\varphi \right \rangle.
\end{flalign*}
We will be interested in the following equations:
\begin{flalign}
\hphantom{A}& \pare{-h^{2}\partial_{x}^{2} + \frac{1}{l^{2}} + \pare{\frac{1}{l^{4}} + \frac{h}{2l^{6}}}x^{2} + h^{2}\frac{ml\pare{ml-1}}{x^{2}}}\psi_{1} = E_{1} \psi_{1}, \label{vp1} \\
& \pare{-h^{2}\partial_{x}^{2} + \frac{1}{l^{2}} + \pare{\frac{1}{l^{4}} + \frac{h}{2l^{6}}}x^{2} + h^{2}\frac{ml\pare{ml+1}}{x^{2}}}\psi_{2} = E_{2} \psi_{2} \label{vp2}.
\end{flalign}
We look first at the equation \eqref{vp1}. We choose:
\begin{equation*}
\psi_{1}\pare{x} = h^{-\frac{1}{4}}\pare{h^{-\frac{1}{2}}x}^{ml} e^{- \pare{\frac{1}{l^{4}} + \frac{h}{2l^{6}}}^{\frac{1}{2}} \frac{x^{2}}{2h}},
\end{equation*}
where the constants are chosen to normalize $\psi_{1}$. Then:
\begin{flalign*}
-h^{2} \partial_{x}^{2} \psi_{1}\pare{x} 
& = -h^{2}\left ( \frac{ml\pare{ml-1}}{x^{2}} \psi_{1}\pare{x} - \frac{ml}{h}\pare{\frac{1}{l^{4}} + \frac{h}{2l^{6}}}^{\frac{1}{2}} \psi_{1}\pare{x} \right . \\
& \quad \left . - \frac{ml+1}{h} \pare{\frac{1}{l^{4}} + \frac{h}{2l^{6}}}^{\frac{1}{2}} \psi_{1}\pare{x}  + \frac{1}{h^{2}}\pare{\frac{1}{l^{4}} + \frac{h}{2l^{6}}}x^{2}\psi_{1}\pare{x} \right ).
\end{flalign*}
Consequently:
\begin{equation*}
 \pare{-h^{2}\partial_{x}^{2} + \frac{1}{l^{2}} + \pare{\frac{1}{l^{4}} + \frac{h}{2l^{6}}}x^{2} + h^{2}\frac{ml\pare{ml-1}}{x^{2}}}\psi_{1} \pare{x}  = E_{1}\pare{h} \psi_{1}\pare{x}
\end{equation*}
with $E_{1}\pare{h} = \frac{1}{l^{2}} +  \pare{2ml+1}\pare{\frac{1}{l^{4}} + \frac{h}{2l^{6}}}^{\frac{1}{2}}h$.\\
We now study equation \eqref{vp2}. We first look for $\alpha$ such that $\alpha\pare{\alpha-1} = ml\pare{ml+1}$. We choose $\alpha_{1} = \frac{1+\sqrt{1+4ml\pare{ml+1}}}{2}$. We then consider:
\begin{equation*}
\psi_{2}\pare{x} = h^{-\frac{1}{4}} \pare{h^{-\frac{1}{2}}x}^{\alpha_{1}} e^{- \pare{\frac{1}{l^{4}}+ \frac{h}{2l^{6}}}^{\frac{1}{2}} \frac{x^{2}}{2h}},
\end{equation*}
where we have also chosen the constant to normalize $\psi_{2}$. In this case, we obtain:
\begin{equation*}
\pare{-h^{2}\partial_{x}^{2} + \frac{1}{l^{2}} + \pare{\frac{1}{l^{4}} + \frac{h}{2l^{6}}}x^{2} + h^{2}\frac{ml\pare{ml+1}}{x^{2}}}\psi_{2} = E_{2}\pare{h} \psi_{2},
\end{equation*}
where $E_{2}\pare{h} = \frac{1}{l^{2}} + \pare{2\alpha_{1}+1}\pare{\frac{1}{l^{4}} + \frac{h}{2l^{6}}}^{\frac{1}{2}} h$.\\
Let us remark that:
\begin{equation*}
\alpha_{1} = \frac{1+\sqrt{1+4ml\pare{ml+1}}}{2} \geq \frac{1+2ml}{2} > ml,
\end{equation*}
so that $E_{2}\pare{h} > E_{1}\pare{h}$. We then see that:
\begin{flalign*}
\hphantom{A} & \prods{\pare{-h^{2}\partial_{x}^{2} + \frac{1}{l^{2}} + \pare{\frac{1}{l^{4}} + \frac{h}{2l^{6}}}x^{2} + h^{2}\pare{\frac{m^{2}l^{2}}{x^{2}} + \begin{pmatrix}
-1 & 0 & 0 & 0 \\
0 & 1 & 0 & 0 \\
0 & 0 & 1 & 0 \\
0 & 0 & 0 & -1
\end{pmatrix} \frac{ml}{x^{2}}}} K^{-1}\varphi}{K^{-1}\varphi} \\
& \leq \prods{\pare{-h^{2}\partial_{x}^{2} + \frac{1}{l^{2}} + \pare{\frac{1}{l^{4}} + \frac{h}{2l^{6}}}x^{2} + h^{2}\pare{\frac{m^{2}l^{2}}{x^{2}} + \begin{pmatrix}
-1 & 0 & 0 & 0 \\
0 & 1 & 0 & 0 \\
0 & 0 & 1 & 0 \\
0 & 0 & 0 & -1
\end{pmatrix} \frac{ml}{x^{2}}}} K^{-1}\varphi}{K^{-1}\varphi} \\
& \quad + \prods{ \pare{E_{2}- E_{1}}\pare{h} \begin{pmatrix}
1 & 0 & 0 & 0 \\
0 & 0 & 0 & 0 \\
0 & 0 & 0 & 0 \\
0 & 0 & 0 & 1
\end{pmatrix}K^{-1}\varphi}{K^{-1}\varphi}.
\end{flalign*}
With $K^{-1}\varphi\pare{x} = \begin{pmatrix}
\psi_{1}\pare{x} \\
\psi_{2}\pare{x} \\
\psi_{2}\pare{x} \\
\psi_{1}\pare{x}
\end{pmatrix}$, we obtain:
\begin{equation*}
\prods{\tilde{P}\varphi}{\varphi} \leq \pare{E_{2}\pare{h}+ \frac{h}{2}} \norme{\varphi}^{2}.
\end{equation*}
The spectral theorem allows us to conclude the proof.
\end{proof}

\subsection{Agmon Estimates}
In this part, we will study the operator $P^{+}= P_{[x_{+},0[}$. We will prove Agmon estimates for this operator. We begin by a lemma: 
\begin{lem}
Let $\phi \in C^{\infty}\pare{[x_{+},0[}$ and $f\in D\pare{P^{+}}$ where the domain is given by \eqref{DomPComp}. Then $e^{-\frac{\phi}{h}} f \in D\pare{P^{+}}$. Moreover, for all $E$, we have:
\begin{flalign*}
\hphantom{A} & \Re \pare{\prods{e^{\frac{\phi}{h}} \pare{-h^{2}\partial_{x}^{2} + V -E} e^{-\frac{\phi}{h}} f}{f}} \\
& = \prods{\pare{-h^{2}\partial_{x}^{2} + A^{2} + h^{2}\pare{ m^{2} B^{2} + m\gamma^{1} B'}- \pare{\phi'}^{2} - E} f}{f} \\
& \quad - 2h \Im\pare{\prods{\begin{pmatrix}
0 & 1 & 0 & 0 \\
0 & 0 & 0 & 0 \\
0 & 0 & 0 & 1 \\
0 & 0 & 0 & 0 
\end{pmatrix} A' f}{f}}.
\end{flalign*}
where $V$ is defined in \eqref{ExprV}.
\end{lem}

\begin{proof}
We have:
\begin{equation*}
V\pare{x} = A^{2}\pare{x} + h^{2}m^{2}B^{2}\pare{x} + h^{2}m \gamma^{1}B'\pare{x} - ih \gamma^{1}\gamma^{2} A'\pare{x}
\end{equation*}
where $A,B,A'$ and $B'$ are real function and 
\begin{equation*}
\gamma^{1}\gamma^{2} =\begin{pmatrix}
0 & -1 & 0 & 0 \\
1 & 0 & 0 & 0 \\
0 & 0 & 0 & -1 \\
0 & 0 & 1 & 0
\end{pmatrix}.
\end{equation*}
Write $f = \begin{pmatrix}
f_{1} \\
f_{2} \\
f_{3} \\
f_{4} 
\end{pmatrix}$. Then:
\begin{equation*}
\prods{ - ih \gamma^{1}\gamma^{2} A'\pare{x} f }{f}  = -2h \pare{\Im\pare{\prods{A'f_{2}}{f_{1}}} + \Im \pare{\prods{A'f_{4}}{f_{3}}}}.
\end{equation*}
This gives the last term of our formula. For the other term, we have:
\begin{equation*}
\prods{-\partial_{x}^{2} \pare{e^{-\frac{\phi}{h}}f}}{e^{\frac{\phi}{h}}f} = \prods{\pare{-\partial_{x}^{2} -\frac{\pare{\phi'}^{2}}{h^{2}}}f}{f} + 2i \Im\pare{\prods{\frac{\phi'}{h} \partial_{x}f}{f}},
\end{equation*}
by integration by part and some calculation similar to the ones in \cite{Gan14}. Taking the real part of this expression gives the result.
\end{proof}
We then have:
\begin{lem} \label{lemexprpsP+}
Let $\varphi \in D\pare{P^{+}}$ and $\chi \in C^{\infty}_{0}\pare{[x_{+},0[}$. Then:
\begin{flalign*}
\Re \prods{e^{\frac{\phi}{h}}\pare{P^{+}-E} \chi \varphi}{e^{\frac{\phi}{h}} \chi \varphi} & = \Re \pare{\prods{e^{\frac{\phi}{h}} \chi \pare{P^{+}-E} \varphi}{e^{\frac{\phi}{h}}\chi \varphi}} \\
& \quad + h^{2} \prods{\varphi}{e^{2\frac{\phi}{h}} \pare{2\frac{\phi'}{h} \chi \chi' + \pare{\chi'}^{2}} \varphi}
\end{flalign*}
\end{lem}

\begin{proof}
First, we have:
\begin{equation*}
\prods{e^{\frac{\phi}{h}}\pare{P^{+}-E} \chi \varphi}{e^{\frac{\phi}{h}} \chi \varphi} = \prods{e^{\frac{\phi}{h}} \chi \pare{P^{+}-E} \varphi}{e^{\frac{\phi}{h}}\chi \varphi} + \prods{\left [ -h^{2}\partial_{x}^{2}, \chi \right ] \varphi}{e^{2\frac{\phi}{h}} \chi \varphi}.
\end{equation*}
The real part of the last term is:
\begin{flalign*}
\Re\pare{\prods{\left [ -h^{2}\partial_{x}^{2}, \chi \right ] \varphi}{e^{2\frac{\phi}{h}} \chi \varphi}} & = \Re \pare{ \prods{-h^{2}\partial_{x}\pare{\chi'\varphi}}{e^{2\frac{\phi}{h}} \chi \varphi}} + \Re\pare{-h^{2} \prods{e^{\frac{\phi}{h}} \chi' \varphi'}{e^{\frac{\phi}{h}} \chi \varphi}} \\
& = \Re\pare{h^{2} \prods{\varphi}{e^{2\frac{\phi}{h}} \pare{2\frac{\phi'}{h} \chi \chi' + \pare{\chi'}^{2}} \varphi}} \\
& \quad + \Re\pare{h^{2} \pare{\prods{e^{\frac{\phi}{h}} \chi \varphi}{e^{\frac{\phi}{h}} \chi' \varphi'}- \overline{\prods{e^{\frac{\phi}{h}} \chi \varphi}{e^{\frac{\phi}{h}} \chi' \varphi'}}}} \\
& = h^{2} \prods{\varphi}{e^{2\frac{\phi}{h}} \pare{2\frac{\phi'}{h} \chi \chi' + \pare{\chi'}^{2}} \varphi},
\end{flalign*}
since the terms are real. This gives the result.
\end{proof}
 Let us introduce $x_{A}\pare{E}$ such that $A^{2}\pare{x_{A}\pare{E}} = E$. Recall that, in a neighborhood of $0$, we have:
\begin{equation*}
A^{2}\pare{x} = \frac{1}{l^{2}} + \frac{x^{2}}{l^{4}} + o\pare{x^{2}},
\end{equation*}
so that:
\begin{equation*}
x_{A}\pare{\frac{1}{l^{2}} + Th} = - l^{2} T^{\frac{1}{2}} h^{\frac{1}{2}} + o \pare{h^{\frac{1}{2}}}.
\end{equation*}
We then have:
\begin{lem}
Let $T>0$ and $\delta > 0$. Then, there exist $k>0$, $h_{0}>0$ such that, for all $\varphi \in D\pare{P^{+}}$, we have:
\begin{equation*}
\prods{\pare{A^{2}\pare{x} - \pare{\frac{1}{l^{2}} + Th} - kx^{2}} \varphi}{\varphi} - 2h \Im\pare{\prods{\begin{pmatrix}
0 & 1 & 0 & 0 \\
0 & 0 & 0 & 0 \\
0 & 0 & 0 & 1 \\
0 & 0 & 0 & 0 
\end{pmatrix} A' \varphi}{\varphi}} > \delta h
\end{equation*}
for all $x \in [x_{+},x_{A}\pare{\frac{1}{l^{2}} + Th}[$ and $h \in (0,h_{0})$.
\end{lem}

\begin{proof}
First, we have:
\begin{equation*}
\abs{2 \Im\pare{\prods{\begin{pmatrix}
0 & 1 & 0 & 0 \\
0 & 0 & 0 & 0 \\
0 & 0 & 0 & 1 \\
0 & 0 & 0 & 0 
\end{pmatrix} A' \varphi}{\varphi}}} \leq \prods{\abs{A'\pare{x}} \varphi}{\varphi}.
\end{equation*}
Thus:
\begin{flalign*}
\hphantom{A} & \prods{\pare{A^{2}\pare{x} - \pare{\frac{1}{l^{2}} + Th} - kx^{2}} \varphi}{\varphi} - 2h \Im\pare{\prods{\begin{pmatrix}
0 & 1 & 0 & 0 \\
0 & 0 & 0 & 0 \\
0 & 0 & 0 & 1 \\
0 & 0 & 0 & 0 
\end{pmatrix} A' \varphi}{\varphi}}\\
&  \geq \prods{\pare{A^{2}\pare{x} - h\abs{A'\pare{x}} - \pare{\frac{1}{l^{2}} + Th} - kx^{2}} \varphi}{\varphi}.
\end{flalign*}
We use the same technique as in \cite{Gan14} to show that the last term is bounded below by $\delta h$. We introduce the function:
\begin{equation*}
M\pare{x,h} = A^{2}\pare{x} +h A'\pare{x} - \pare{\frac{1}{l^{2}} + Th} - kx^{2}
\end{equation*}
since $A'\pare{x} <0$ on $[x_{+},0[$ (by increasing $x_{+}$ if needed). Choose $ k = \frac{\delta}{4 l^{4} \pare{T+ 2\delta}}$, then:
\begin{equation*}
M\pare{x_{A}\pare{\frac{1}{l^{2}} + \pare{T+2\delta}h},h} 
 = 2\delta h - \frac{\delta}{4} h + o\pare{h} > \delta h,
\end{equation*}
for $h$ sufficiently small. Note that $k \leq \frac{1}{8l^{4}}$. We have:
\begin{equation*}
M'\pare{x,h} = \pare{A^{2}}'\pare{x} + h A''\pare{x} - 2kx,
\end{equation*}
where $ \pare{A^{2}}'\pare{x} = \frac{2x}{l^{4}} + o\pare{x}$ and $A''\pare{x} = \frac{1}{l^{3}} + o\pare{1}$. We obtain:
\begin{flalign*}
M'\pare{x_{A}\pare{\frac{1}{l^{2}} + \pare{T+2\delta}h},h} & \leq \frac{6}{4l^{4}} x_{A}\pare{\frac{1}{l^{2}} + \pare{T+2\delta}h} + o\pare{h^{\frac{1}{2}}} \\
& < -\frac{5}{4l^{4}} l^{2}T^{\frac{1}{2}} h^{\frac{1}{2}} <0.
\end{flalign*}
Moreover:
\begin{equation*}
M''\pare{x,h} = \pare{A^{2}}''\pare{x} + h A^{(3)}\pare{x} - 2k,
\end{equation*}
where $\pare{A^{2}}''\pare{x} = \frac{2}{l^{4}} + o\pare{1}$, $A^{(3)}\pare{x} = O\pare{1}$ and $2k \leq \frac{1}{4l^{4}}$. Thus:
\begin{equation*}
M''\pare{x,h} \geq \frac{7}{4l^{4}} + o\pare{1} + h O\pare{1}.
\end{equation*}
Since $\frac{7}{4l^{4}} + o\pare{1} >0$ on $[-2A,0]$, we have:
\begin{equation*}
M''\pare{x,h} \geq 0,
\end{equation*}
for $h$ sufficiently small, on $[-A,0]$. Hence, on $[-A,x_{A}\pare{\frac{1}{l^{2}} + \pare{T+2\delta}h}]$, $M'<0$ and:
\begin{equation*}
M\pare{x,h} \geq \delta h.
\end{equation*}
On $[x_{+},-A]$, we have:
\begin{equation*}
A^{2}\pare{x} - \pare{\frac{1}{l^{2}} + Th} - k x^{2} \geq A^{2}\pare{x} - \pare{\frac{1}{l^{2}} + Th} - k x_{+}^{2} > C_{1}
\end{equation*}
where we can decrease $k$ if needed. For $h$ sufficiently small, we have:
\begin{equation*}
A^{2}\pare{x} +h A'\pare{x} - \pare{\frac{1}{l^{2}} + Th} - k x^{2}> C_{1} - h C_{A},
\end{equation*}
where $C_{A} = \underset{[x_{+},-A]}{\sup} A'\pare{x}$. Thus:
\begin{equation*}
M\pare{x,h} \geq C_{2}
\end{equation*}
on $[x_{+},-A]$. The constant $C_{2}$ have to satisfy:
\begin{equation*}
C_{2} \geq \delta h,
\end{equation*}
that is $C_{1} > \pare{\delta + C_{A}}h$. For a good choice of $k$ and $h_{0}$, we obtain the result.
\end{proof}
We can now prove the following estimate:
\begin{prop}
Let $T>0$. There exist constants $h_{0} >0$, $C>0$, and $c >0$ such that:
\begin{equation*}
\norme{e^\frac{x^{2}}{ch} \varphi}_{L^{2}\pare{[x_{+},0[}} \leq C\pare{\norme{\varphi}_{L^{2}\pare{[x_{+},0[}} + h^{-1} \norme{e^{\frac{x^{2}}{ch}} \pare{P^{+} - E\pare{h}}\varphi}_{L^{2}\pare{[x_{+},0[}}}
\end{equation*}
for all $h\in (0, h_{0})$, $\varphi \in D\pare{P^{+}}$ and $E\pare{h}$ such that $E\pare{h} < 1+Th$.
\end{prop}

\begin{proof}
Fix $\delta >0$. By the preceding lemma, we can choose a constant $c>0$ such that, if $\phi\pare{x} = \frac{x^{2}}{c}$, we have:
\begin{equation*}
\prods{\pare{A^{2}\pare{x} - \pare{\frac{1}{l^{2}} + Th} - \pare{\phi'}^{2}} \varphi}{\varphi} - 2h \Im\pare{\prods{\begin{pmatrix}
0 & 1 & 0 & 0 \\
0 & 0 & 0 & 0 \\
0 & 0 & 0 & 1 \\
0 & 0 & 0 & 0 
\end{pmatrix} A' \varphi}{\varphi}} > \delta h
\end{equation*}
Write $X_{1}= - l^{2} \pare{T+2\delta}^{\frac{1}{2}} - \epsilon$ and $X_{2} = - l^{2} \pare{T+2\delta}^{\frac{1}{2}} - 2\epsilon$. Let $\zeta \in C^{\infty}$ such that $\zeta \equiv 0$ on $[X_{1},0]$ and $\zeta \equiv 1$ on $[h^{-\frac{1}{2}}x_{+}, X_{2}]$. Let $\chi\pare{x,h} = \zeta\pare{h^{-\frac{1}{2}}x;h}$ such that $\chi$ has support in $[x_{+},x_{A}\pare{\frac{1}{l^{2}} + \pare{T+2\delta}h}]$. Choose $f = e^{\frac{\phi}{h}} \chi \varphi$. By the preceding lemma, we have: 
\begin{flalign*}
\hphantom{A} & \Re\pare{\prods{e^{\frac{\phi}{h}}\pare{-h^{2}\partial_{x}^{2} + V -E} \chi \varphi}{e^{\frac{\phi}{h}}\chi \varphi}} = \\
&  \prods{\pare{-h^{2}\partial_{x}^{2} + A^{2} + h^{2}m^{2} B^{2} + h^{2}m\gamma^{1}B'} e^{\frac{\phi}{h}} \chi \varphi}{e^{\frac{\phi}{h}} \chi \varphi} \quad +\prods{\pare{ - E - \pare{\phi'}^{2}} e^{\frac{\phi}{h}} \chi \varphi}{e^{\frac{\phi}{h}} \chi \varphi}\\
& \quad - 2h \Im\pare{\prods{\begin{pmatrix}
0 & 1 & 0 & 0 \\
0 & 0 & 0 & 0 \\
0 & 0 & 0 & 1 \\
0 & 0 & 0 & 0 
\end{pmatrix} A' e^{\frac{\phi}{h}} \chi \varphi}{e^{\frac{\phi}{h}} \chi \varphi}} 
 \geq \delta h \norme{e^{\frac{\phi}{h}} \chi \varphi}^{2}
\end{flalign*}
Then, by lemma \ref{lemexprpsP+}, we have:
\begin{flalign*}
\Re \prods{e^{\frac{\phi}{h}}\pare{P^{+}-E} \chi \varphi}{e^{\frac{\phi}{h}} \chi \varphi}_{L^{2}\pare{[x_{+},0]}} & = \Re \pare{\prods{e^{\frac{\phi}{h}} \chi \pare{P^{+}-E} \varphi}{e^{\frac{\phi}{h}}\chi \varphi}_{L^{2}\pare{[x_{+},0]}}} \\
& \quad + h^{2} \prods{\varphi}{e^{2\frac{\phi}{h}} \pare{2\frac{\phi'}{h} \chi \chi' + \pare{\chi'}^{2}} \varphi}_{L^{2}\pare{[x_{+},0]}} \\
& \leq h^{2} \prods{\varphi}{e^{2\frac{\phi}{h}} \pare{2\frac{\phi'}{h} \chi \chi' + \pare{\chi'}^{2}} \varphi}_{L^{2}\pare{[x_{+},0]}}\\
& \quad + \norme{e^{\frac{\phi}{h}} \chi \pare{P^{+}-E} \varphi}_{L^{2}\pare{[x_{+},0]}}\norme{e^{\frac{\phi}{h}}\chi \varphi}_{L^{2}\pare{[x_{+},0]}}.
\end{flalign*}
Thus:
\begin{flalign*}
\delta h \norme{e^{\frac{\phi}{h}} \chi \varphi}^{2}_{L^{2}\pare{[x_{+},0]}} & \leq h^{2} \prods{\varphi}{e^{2\frac{\phi}{h}} \pare{2\frac{\phi'}{h} \chi \chi' + \pare{\chi'}^{2}} \varphi}_{L^{2}\pare{[x_{+},0]}}\\
& \quad + \norme{e^{\frac{\phi}{h}} \chi \pare{P^{+}-E} \varphi}_{L^{2}\pare{[x_{+},0]}}\norme{e^{\frac{\phi}{h}}\chi \varphi}_{L^{2}\pare{[x_{+},0]}}.
\end{flalign*}
By the same argument as in \cite{Gan14} proposition $3.12$, we obtain:
\begin{flalign*}
\norme{e^{\frac{\phi}{h}} \chi \varphi}^{2}_{L^{2}\pare{[x_{+},0]}} & \leq 2\delta^{-1}h \prods{\varphi}{e^{2\frac{\phi}{h}} \pare{2\frac{\phi'}{h} \chi \chi' + \pare{\chi'}^{2}} \varphi}_{L^{2}\pare{[x_{+},0]}} \\
& \quad + \pare{\delta h}^{-2} \norme{e^{\frac{\phi}{h}} \chi \pare{P^{+}-E} \varphi}^{2}_{L^{2}\pare{[x_{+},0]}}.
\end{flalign*}
Remark that $\chi'$ is non zero when $x\in h^{\frac{1}{2}}[X_{2},X_{1}]$ so that
\begin{equation*}
\underset{\text{supp}\pare{\chi'}}{\sup} e^{\frac{\phi}{h}} = O\pare{1}, \hspace{3mm} \underset{\text{supp}\pare{\chi'}}{\sup} \abs{\phi'} = O\pare{h^{\frac{1}{2}}}.
\end{equation*}
We also have:
\begin{equation*}
\chi'\pare{x} = h^{-\frac{1}{2}} \zeta'\pare{h^{-\frac{1}{2}}x;h},
\end{equation*}
which imply $\sup \abs{\chi'\pare{x}} = O\pare{h^{-\frac{1}{2}}}$ and $\sup \abs{\frac{\phi'}{h}\chi \chi'} = O\pare{h^{-1}}$. Consequently:
\begin{equation*}
\norme{e^{\frac{\phi}{h}}\varphi}_{L^{2}\pare{[x_{+}, h^{\frac{1}{2}}X_{2}]}} \leq C_{1} \norme{\varphi}_{L^{2}\pare{[x_{+},0]}} + C_{2}h^{-1} \norme{e^{\frac{\phi}{h}} \pare{P^{+}-E} \varphi}_{L^{2}\pare{[x_{+},0]}}.
\end{equation*}
The result then follow from the fact that:
\begin{equation*}
\norme{e^{\frac{\phi}{h}} \varphi}_{L^{2}\pare{[h^{\frac{1}{2}}X_{2},0]}} \leq C_{3} \norme{\varphi}_{L^{2}\pare{[x_{+},0]}}.
\end{equation*}
\end{proof}
Let $S>0$. We choose $A_{1}<A_{2} < x_{A}\pare{\frac{1}{l^{2}}+S}$ and introduce $\Sigma_{i} = [x_{+},A_{i}[$ such that $\Sigma_{1} \Subset \Sigma_{2} \Subset [x_{+},x_{A}\pare{\frac{1}{l^{2}}+S}]$.
\begin{prop}
Now, let $S>0$ such that $\frac{1}{l^{2}}+S < A^{2}\pare{x_{+}}$. There exist constants $h_{0}>0$, $C>0$, and $\epsilon >0$ such that:
\begin{equation*}
\norme{\varphi}_{L^{2}\pare{\Sigma_{1}}} \leq C\pare{ e^{-\frac{\epsilon}{h}} \norme{\varphi}_{L^{2}\pare{\Sigma_{2}}} + \norme{\pare{P^{+}-E\pare{h}}\varphi}_{L^{2}\pare{\Sigma_{2}}}},
\end{equation*}
for all $h\in (0,h_{0})$, $\varphi \in D\pare{P^{+}}$, and for all $E\pare{h}$ such that $E\pare{h} < \frac{1}{l^{2}}+S$.
\end{prop}

\begin{proof}
For $\delta$ sufficiently small, we can suppose that $\Sigma_{2} \Subset [x_{+}, x_{A}\pare{\frac{1}{l^{2}}+ S+2\delta}]$. Let $\chi_{1}$ such that $\chi_{1} \equiv 1$ on $\Sigma_{1}$ and $\text{supp}\pare{\chi_{1}} \subseteq \Sigma_{2}$. Let $\chi_{2}$ such that $\chi_{2} \equiv 1$ on $\text{supp}\pare{\chi_{1}}$ and $\text{supp}\pare{\chi_{2}} \subseteq \Sigma_{2}$. Then, on $\Sigma_{2}$, $A^{2} - \pare{\frac{1}{l^{2}}+S} > \delta$. For $h$ and $\epsilon$ sufficiently small, we also have:
\begin{equation*}
A^{2}- h\abs{A'} - \pare{\frac{1}{l^{2}}+S} - \epsilon^{2} \pare{\chi_{1}'}^{2}> \delta.
\end{equation*}
Using the same reasoning as in the preceding proof, we obtain:
\begin{flalign*}
\delta  \norme{e^{\frac{\phi}{h}} \chi_{2} \varphi}^{2}_{L^{2}\pare{[x_{+},0]}} & \leq h^{2} \prods{\varphi}{e^{2\frac{\phi}{h}} \pare{2\frac{\epsilon \chi_{1}'}{h} \chi_{2} \chi_{2}' + \pare{\chi_{2}'}^{2}} \varphi}_{L^{2}\pare{[x_{+},0]}}\\
& \quad + \norme{e^{\frac{\phi}{h}} \chi_{2} \pare{P^{+}-E} \varphi}_{L^{2}\pare{[x_{+},0]}}\norme{e^{\frac{\phi}{h}}\chi_{2} \varphi}_{L^{2}\pare{[x_{+},0]}},
\end{flalign*}
where we have taken $\phi \pare{x} = \epsilon \chi_{1}\pare{x}$. As before, we obtain:
\begin{flalign*}
\norme{e^{\frac{\phi}{h}} \chi_{2} \varphi}^{2}_{L^{2}\pare{[x_{+},0]}} & \leq 2\delta^{-1}h^{2} \prods{\varphi}{e^{2\frac{\phi}{h}} \pare{2\frac{\epsilon \chi_{1}'}{h} \chi_{2} \chi_{2}' + \pare{\chi_{2}'}^{2}} \varphi}_{L^{2}\pare{[x_{+},0]}} \\
& \quad + \pare{\delta}^{-2} \norme{e^{\frac{\phi}{h}} \chi_{2} \pare{P^{+}-E} \varphi}^{2}_{L^{2}\pare{[x_{+},0]}}
\end{flalign*}
Thanks to the choice of $\chi_{1}$ and $\chi_{2}$, we have:
\begin{equation*}
 \delta \norme{e^{\frac{\epsilon \chi_{1}}{h}} \chi_{2} \varphi}_{L^{2}\pare{\Sigma_{1}}} = \delta \norme{e^{\frac{\epsilon}{h}} \varphi}_{L^{2}\pare{\Sigma_{1}}}, \hspace{5mm} \delta \norme{e^{\frac{\epsilon}{h}} \varphi}_{L^{2}\pare{\Sigma_{1}}} \leq \delta \norme{e^{\frac{\phi}{h}} \chi_{2}\varphi}_{L^{2}\pare{[x_{+},0]}}.
\end{equation*}
Since $\chi_{1} \equiv 0$ on the support of $\chi_{2}'$ and $\epsilon \chi_{1} \leq \epsilon$, we obtain:
\begin{equation*}
\norme{e^{\frac{\epsilon}{h}} \varphi}_{L^{2}\pare{\Sigma_{1}}} \leq 2\delta^{-1}h^{2} \sup \pare{\pare{\chi_{2}'}^{2}} \norme{\varphi}_{L^{2}\pare{\Sigma_{2}}} + \delta^{-2} e^{\frac{\epsilon}{h}} \norme{\pare{P^{+} - E} \varphi}_{L^{2}\pare{\Sigma_{2}}}
\end{equation*}
which gives the result when multiplying by $e^{-\frac{\epsilon}{h}}$.
\end{proof}

\subsection{Estimates for the eigenvalues of \texorpdfstring{$P^{+}$}{P+}}

In this section, we prove that it is possible to find an eigenvalue of the operator $P^{+}$ at distance $O\pare{h^{\frac{1}{2}}}$ from the eigenvalue of $\tilde{P}$ obtained in \ref{VpPtilde}. We have:
\begin{prop}
Let $T>0$. There exists $h_{0}>0$ such that for all $h\in (0,h_{0})$, there exists an eigenvalue $E^{+}\pare{h}$ of $P^{+}$ satisfying:
\begin{equation*}
\abs{E^{+}\pare{h} - E_{2}\pare{h}} <Ch^{\frac{1}{2}},
\end{equation*}
where $E_{2}\pare{h} = \frac{1}{l^{2}} + \pare{2\alpha_{1}+1}\pare{\frac{1}{l^{4}} + \frac{h}{2l^{6}}}^{\frac{1}{2}} h$
\end{prop}

\begin{proof}
We introduce:
\begin{flalign*}
P^{\#} & = -h^{2}\partial_{x}^{2} + \frac{1}{l^{2}} + \pare{\frac{1}{l^{4}} + \frac{h}{2l^{6}}}x^{2} + h^{2}\pare{\frac{m^{2}l^{2}}{x^{2}} + i\gamma^{1} \frac{ml}{x^{2}}} \\
& \quad + \pare{E_{2}- E_{1}}\pare{h} K\begin{pmatrix}
1 & 0 & 0 & 0 \\
0 & 0 & 0 & 0 \\
0 & 0 & 0 & 0 \\
0 & 0 & 0 & 1
\end{pmatrix}K^{-1}
\end{flalign*}
where $E_{1}\pare{h} = \frac{1}{l^{2}} +  \pare{2ml+1}\pare{\frac{1}{l^{4}} + \frac{h}{2l^{6}}}^{\frac{1}{2}}h$ such that $\pare{E_{2} - E_{1}}\pare{h} >0$ and \\ $K= \frac{1}{\sqrt{2}} \begin{pmatrix}
1 & 0 & 1 & 0 \\
0 & 1 & 0 & 1 \\
1 & 0 & -1 & 0 \\
0 & 1 & 0 & -1
\end{pmatrix}$. This operator admit $E_{2}\pare{h}$ as an eigenvalue. Denote $\varphi_{2}$ an associated eigenvector. Let $S>0$ such that $\frac{1}{l^{2}}+S < A^{2}\pare{x_{+}}$ and $\mathcal{A}<x_{A}\pare{\frac{1}{l^{2}}+S}$. Let $\chi \in C^{\infty}$ such that $\chi \equiv 1$ on $[\mathcal{A},0[$ and $\text{supp} \pare{\chi} = [x_{+},0]$. Then $\chi \varphi_{2} \in D\pare{P^{+}}$ and:
\begin{equation*}
\pare{P^{+} - E_{2}\pare{h}}\pare{\chi \varphi_{2}} = \pare{P^{+} - P^{\#}}\pare{\chi \varphi_{2}} + \pare{P^{\#} - E_{2}\pare{h}} \pare{\chi \varphi_{2}} = \chi R\pare{h} \varphi_{2} + \left [ -h^{2}\partial_{x}^{2}, \chi \right ] \varphi_{2},
\end{equation*}
where:
\begin{flalign*}
R\pare{h} & = A^{2}\pare{x} - \pare{\frac{1}{l^{2}} + \frac{x^{2}}{l^{4}}} - h \frac{x^{2}}{l^{6}} - ih \gamma^{1}\gamma^{2} A'\pare{x} + h^{2}m^{2}\pare{B^{2}\pare{x} - \frac{l^{2}}{x^{2}}} \\
& \quad + ih^{2} m\gamma^{1}\pare{B'\pare{x} - \frac{l}{x^{2}}} + \pare{E_{2}-E_{1}}\pare{h}K\begin{pmatrix}
1 & 0 & 0 & 0 \\
0 & 0 & 0 & 0 \\
0 & 0 & 0 & 0 \\
0 & 0 & 0 & 1
\end{pmatrix}K^{-1} \\
& = o\pare{x^{2}} + h O\pare{x^{2}} + h O\pare{x} + O\pare{h^{2}} + O\pare{h}.
\end{flalign*}
Recall that $K^{-1}\varphi_{2}\pare{x} = \begin{pmatrix}
\psi_{1}\pare{x} \\
\psi_{2}\pare{x} \\
\psi_{2}\pare{x} \\
\psi_{1}\pare{x}
\end{pmatrix}$ where $\psi_{1}\pare{x} = h^{-\frac{1}{4}} \pare{h^{-\frac{1}{2}}x}^{ml} e^{- \pare{\frac{1}{l^{4}}+ \frac{h}{2l^{6}}}^{\frac{1}{2}} \frac{x^{2}}{2h}}$ and $\psi_{2}\pare{x} = h^{-\frac{1}{4}} \pare{h^{-\frac{1}{2}}x}^{\alpha_{1}} e^{- \pare{\frac{1}{l^{4}}+ \frac{h}{2l^{6}}}^{\frac{1}{2}} \frac{x^{2}}{2h}}$ with $\alpha_{1} = \frac{1+\sqrt{1+4ml\pare{ml+1}}}{2}$. Using the change of variable $y = h^{-\frac{1}{2}}x$ and the exponential decay of our eigenvector, we see that:
\begin{equation*}
\norme{\chi R\pare{h} \varphi_{2}}_{L^{2}\pare{[x_{+},0]}} = O\pare{h^{\frac{1}{2}}}.
\end{equation*}
The commutator is equal to:
\begin{equation*}
\left [-h^{2}\partial_{x}^{2}, \chi \right ] = -h^{2} \chi'' - h^{2} \chi' \partial_{x}.
\end{equation*}
The support of $\chi'$ and $\chi''$ and the formula for $\varphi_{2}$ give $\epsilon >0$ such that:
\begin{equation*}
\norme{\left [-h^{2}\partial_{x}^{2}, \chi \right ] \varphi_{2}}_{L^{2}\pare{[x_{+},0]}} = O\pare{e^{-\frac{ \epsilon}{h}}}.
\end{equation*}
Consequently, we have:
\begin{equation*}
\norme{\pare{P^{+} - E_{2}\pare{h}}\pare{\chi \varphi_{2}}}_{L^{2}\pare{[x_{+},0]}} = O\pare{h^{\frac{1}{2}}}.
\end{equation*}
Since $\norme{\chi \varphi_{2}}_{L^{2}\pare{[x_{+},0]}} = 1 - O\pare{e^{-\frac{d}{h}}}$, we can normalize $\chi\varphi_{2}$. The spectral theorem gives an eigenvalue $E^{+}\pare{h}$ of $P^{+}$ such that:
\begin{equation*}
\abs{E^{+}\pare{h} - E_{2}\pare{h}} \leq C h^{\frac{1}{2}}
\end{equation*}
for a certain constant $C>0$.
\end{proof}

\subsection{Construction of quasimodes}

We now introduce the operator $H^{+}$ which is $H$ on $[x_{+},0[$ with domain:
\begin{equation*}
D\pare{H^{+}} = \{\varphi \in L^{2}\pare{[x_{+},0[} \lvert H^{+}\varphi \in L^{2}\pare{[x_{+},0[}, \varphi\pare{x_{+}}=0 \}.
\end{equation*}
We have $P^{+} = \pare{H^{+}}^{2}$ and $\pare{E^{+}}^{2}$ is an eigenvalue of $P^{+}$ if and only if $\pm E^{+}$ is an eigenvalue of $H^{+}$. We then have the:
\begin{theo}
Let $S>0$ such that $\frac{1}{l^{2}}+S < A^{2}\pare{x_{+}}$. There exist constants $h_{0} >0$, $D >0$, a real number $\pare{E^{+}\pare{h}}^{\frac{1}{2}}>0$ such that $E^{+}\pare{h} < \frac{1}{l^{2}} + S$ and a function $\varphi_{h} \in D\pare{H}$ such that $\norme{\varphi_{h}}_{L^{2}\pare{]-\infty,0[}} = 1$ satisfying:
\begin{equation*}
\norme{\pare{H - \pare{E^{+}\pare{h}}^{\frac{1}{2}}}\varphi_{h}}_{L^{2}\pare{]-\infty,0[}} \leq e^{-\frac{D}{h}},
\end{equation*}
for all $h \in (0,h_{0})$.
\end{theo}

\begin{proof}
Let $A< x_{A}\pare{\pare{\frac{1}{l^{2}}+S}}$ and $\chi \in C^{\infty}_{0}$ such that $\chi \equiv 1$ on $[A,0[$ and $\text{supp}\pare{\chi} = [x_{+},0[$. Denote $\varphi^{+}_{h}$ an eigenvector of $H^{+}$ for the eigenvalue $\pare{E^{+}\pare{h}}^{\frac{1}{2}}$ such that $E^{+}\pare{h} <\frac{1}{l^{2}} + S$ which exists by the last section. We use $\varphi_{h} = \chi \varphi^{+}_{h}$ that we extend by $0$ outside $[x_{+},0[$ so that $\varphi_{h} \in D\pare{H}$ and:
\begin{equation*}
\norme{\pare{H - \pare{E^{+}\pare{h}}^{\frac{1}{2}}} \varphi_{h}}_{L^{2}\pare{]-\infty,0[}} = \norme{\left [ i\gamma^{0}\gamma^{1} h\partial_{x}, \chi \right ] \varphi^{+}_{h}}_{L^{2}\pare{]-\infty,0[}}.
\end{equation*}
Using Agmon estimates, we obtain:
\begin{equation*}
\norme{\left [ i\gamma^{0}\gamma^{1} h\partial_{x}, \chi \right ]\varphi^{+}_{h}} = \norme{i\gamma^{0}\gamma^{1} h \chi' \varphi^{+}_{h}} \leq C e^{-\frac{\epsilon}{h}} \norme{\varphi^{+}_{h}}.
\end{equation*}
This gives:
\begin{equation*}
\norme{\pare{H - \pare{E^{+}\pare{h}}^{\frac{1}{2}}} \varphi_{h}}_{L^{2}\pare{]-\infty,0[}} \leq e^{-\frac{D}{h}}.
\end{equation*}
Since we can normalize $\varphi_{h}$, the result is proven.
\end{proof}

\section{Lower bound on the local energy}

In this section, we will make use of an argument in \cite{HolSmu3} to prove the:
\begin{theo}
For all compact $K \subset ]-\infty,0[$, there exists a constant $C>0$ such that:
\begin{equation*}
\underset{t \to +\infty}{\limsup} \underset{\varphi \in \mathcal{H}, \norme{\varphi}= 1}{\sup} \ln\pare{t}\norme{e^{itH}\varphi}_{L^{2}\pare{K}} \geq C.
\end{equation*}
\end{theo}

\begin{proof}
Let $\varphi_{h}$ be a quasimode as in the precedent section. We introduce the function $\psi_{h} = e^{iE^{+}\pare{h}t}\varphi_{h}$. Then:
\begin{equation*}
-i \partial_{t} \psi_{h} - H \psi_{h} = - \pare{H-E^{+}\pare{h}}\psi_{h}
\end{equation*}
and $\norme{\pare{H-E^{+}\pare{h}}\psi_{h}}_{L^{2}\pare{]-\infty,0[}} \leq e^{-\frac{D}{h}}$. We consider the preceding equation as non homogeneous equation of the form:
\begin{equation*}
-i \partial_{t} \chi - H \chi = F
\end{equation*}
where $\norme{F}_{L^{2}\pare{]-\infty,0[}} \leq e^{-\frac{D}{h}}$. Let $\tilde{\psi}_{h} = e^{itH}\varphi_{h}$ a solution of the homogeneous equation. By Duhamel formula, we have:
\begin{equation*}
\chi = e^{itH}\varphi_{h} + \int_{0}^{t} e^{i\pare{t-s}H}\varphi_{h} F ds.
\end{equation*}
Consequently:
\begin{flalign*}
\abs{\norme{\psi_{h}}_{L^{2}\pare{K}} - \norme{ \tilde{\psi}_{h}}_{L^{2}\pare{K}}} & \leq \norme{\psi_{h} - \tilde{\psi}_{h}}_{L^{2}\pare{K}}\\
& \leq \int_{0}^{t} \norme{e^{i\pare{t-s}H}\varphi_{h} F}_{L^{2}\pare{K}} ds \\
& \leq Ct \norme{F} \leq Ct e^{-\frac{D}{h}}.
\end{flalign*}
We deduce that:
\begin{equation*}
\norme{e^{itH}\varphi_{h}}_{L^{2}\pare{K}} \geq \norme{\psi_{h}}_{L^{2}\pare{K}} - Ct e^{-\frac{D}{h}} = \norme{\varphi_{h}}_{L^{2}\pare{K}}- Ct e^{-\frac{D}{h}}.
\end{equation*}
Let $\lambda= \norme{\varphi_{h}}_{L^{2}\pare{K}}$. Take $t_{h} = \frac{\lambda-h}{C}e^{\frac{D}{h}}$ such that $\lambda - Ct_{h} e^{-\frac{D}{h}} = h$. We then have:
\begin{equation*}
h = \frac{D}{\ln\pare{t_{h}}} + \frac{h\ln\pare{\frac{\lambda-h}{C}}}{\ln\pare{t_{h}}}.
\end{equation*}
For $h$ sufficiently small, we have $\abs{h\ln\pare{\frac{\lambda-h}{C}}} \leq \frac{D}{2}$ which gives:
\begin{equation*}
\norme{e^{it_{h}H}\varphi_{h}}_{L^{2}\pare{K}} \geq \frac{D}{2\ln\pare{t_{h}}}.
\end{equation*}
When $h$ goes to $0$, $t_{h}$ goes to $+ \infty$ and we have:
\begin{equation*}
\underset{t \to +\infty}{\limsup} \underset{\varphi \in \mathcal{H}, \norme{\varphi}= 1}{\sup} \ln\pare{t} \norme{e^{itH}\varphi}_{L^{2}\pare{K}} \geq \frac{D}{2} >0
\end{equation*}
\end{proof}

\bibliographystyle{plain-fr}
\bibliography{BiblioQuasimodes}
\nocite{*}

\end{document}